\documentclass[a4paper,UKenglish,cleveref, autoref, thm-restate]{lipics-v2021}

\usepackage{amsmath}
\usepackage{amssymb}
\usepackage{verbatim}
\usepackage{stmaryrd}

\newcommand{\lenc}{\mathrm{lenc}}
\newcommand{\gra}{\mathrm{GRA}}
\newcommand{\subf}{\mathrm{Subf}}
\newcommand{\E}{\langle \mathrm{E} \rangle}
\newcommand{\A}{\langle \mathrm{A} \rangle}
\newcommand{\PML}{\mathrm{PML}}

\title{Complexity of Polyadic Boolean Modal Logics: \ \\ Model Checking and Satisfiability} 

\titlerunning{Complexity of Polyadic Boolean Modal Logics} 

\author{Reijo Jaakkola}{Tampere University, Finland \and \url{https://reijojaakkola.github.io/} }{reijo.jaakkola@tuni.fi}{https://orcid.org/0000-0003-4714-4637}{}

\authorrunning{R. Jaakkola} 

\Copyright{Reijo Jaakkola} 

{

} 


\category{} 

\relatedversion{} 

\acknowledgements{I want to thank Antti Kuusisto for suggesting this topic to me and for discussing it with me on numerous occasions. I also want to thank the anonymous reviewers for their valuable comments which greatly improved the presentation of this paper.}

\nolinenumbers 

\hideLIPIcs  

\EventEditors{John Q. Open and Joan R. Access}
\EventNoEds{2}
\EventLongTitle{42nd Conference on Very Important Topics (CVIT 2016)}
\EventShortTitle{CVIT 2016}
\EventAcronym{CVIT}
\EventYear{2016}
\EventDate{December 24--27, 2016}
\EventLocation{Little Whinging, United Kingdom}
\EventLogo{}
\SeriesVolume{42}
\ArticleNo{23}

\begin{document}

\maketitle

\begin{abstract}
We study the computational complexity of model checking and satisfiability problems of polyadic modal logics extended with permutations and Boolean operators on accessibility relations. First, we show that the combined complexity of the model checking problem for the resulting logic is \textsc{PTime}-complete. Secondly, we show that the satisfiability problem of polyadic modal logic extended with negation on accessibility relations is \textsc{ExpTime}-complete. Finally, we show that the satisfiability problem of polyadic modal logic with permutations and Boolean operators on accessibility relations is \textsc{ExpTime}-complete, under the necessary assumption that the number of accessibility relations that can be used is bounded by a constant. 
\end{abstract}

\section{Introduction}\label{sec:introduction}

In recent years there has been increasing interest in generalizing complexity results for logics that only have access to relation symbols of arity at most two, to logics that have access to relations of arbitrary high arity, which we will call their \emph{polyadic} extensions, see for example \cite{Bednarczyk2021ExploitingFS,OneDimensionalFragment,IsoTuisku2018Ordered,IsoTuisku2021DescriptionLA,ORDEREDFRAGMENTS,Kieronski2014ComplexityAE,10.1093/logcom/exac002,10.1145/3209108.3209168}. In \cite{OneDimensionalFragment} the authors introduced the uniform one-dimensional logic $\mathrm{U}_1$, which is a polyadic extension of the two-variable logic $\mathrm{FO}^2$. It was proved in \cite{Kieronski2014ComplexityAE} that the complexity of $\mathrm{U}_1$ satisfiability problem is \textsc{NExpTime}-complete, which is the same as for $\mathrm{FO}^2$ \cite{GKV97}. In the very recent work \cite{Bednarczyk2021ExploitingFS} the forward guarded fragment FGF was introduced, which is a polyadic extension of the description logic $\mathcal{ALC}$ with global diamond. In the same work it was established that its satisfiability and conjunctive query entailment problems have the same complexity as the corresponding problems in the case of $\mathcal{ALC}$, i.e., both are \textsc{ExpTime}-complete. In the two aforementioned examples the complexities of the base logic and its polyadic extension coincided, but there are also examples where the polyadic extension has a higher complexity. For instance, in \cite{Kieronski19} it is proved that the satisfiability problem of guarded $\mathrm{U}_1$, which is a polyadic extension of guarded $\mathrm{FO}^2$, is \textsc{NExpTime}-complete, while the satisfiability problem of guarded $\mathrm{FO}^2$ is \textsc{ExpTime}-complete \cite{Gradel99}.

While lifting complexity results from base logics to their polyadic extensions is, at least to the author, already intrinsically interesting, it also has more applied motivations stemming from, say, database theory, since polyadic relations occur naturally in various contexts and having access to them can be advantageous. For instance, a simple relation such as ``Alice received a message M from Bob'' is best viewed as a ternary relation. One application area where the potential need for polyadic logics has been recognized is the very active field of \emph{description logics} \cite{DescriptionLogicBook}. Indeed, given that several description logics used in knowledge representation can only access binary roles, it should not come as a surprise that several polyadic \emph{description logics} have been already suggested in the literature \cite{CalvaneseGL98,LutzST99,Schmolze89}. Finding polyadic extensions of modal logics can be seen as attempts at finding natural polyadic description logics.

To make the research around polyadic extensions more systematic, in \cite{IsoTuisku2021DescriptionLA} the authors outlined a research program for systematically extending complexity results for description logics --- and modal logics more generally --- into their polyadic counterparts, following the ideas presented in \cite{KuusistoUfSurvey}. The main idea is the following: since description logics can be seen as standard modal logic extended with \emph{relation operators} -- such as inverse roles and counting -- and the standard modal logic has a canonical extension into a polyadic modal logic (namely the extension in which diamonds can bind multiple formulas), one can easily obtain quite canonical extensions of known description logics\footnote{Of course, some conceptual work is needed to figure out what are the canonical polyadic extensions of the relevant role operators.}. To demonstrate their research program in action, the polyadic extension of $\mathcal{ALCQI}$, i.e., the extension of $\mathcal{ALC}$ with inverse roles and counting, was studied in \cite{IsoTuisku2021DescriptionLA}. The authors proved that the concept satisfiability problem for this logic is \textsc{Pspace}-complete, which is the same as in the case of $\mathcal{ALCQI}$.

The main purpose of the present work is to contribute to the above research program in the context of Boolean modal logics, which are modal logics extended with Boolean operators on accessibility relations \cite{Lutz2000TheCO}. In these logics one can for example write down formulas such as $\langle \neg R \rangle p$, which expresses that one can reach a world in which $p$ is true via the complement of the accessibility relation $R$. While being very natural modal logics as such, they are also closely related to other interesting logics such as modal logics extended with the window operator \cite{windowmodality,Lutz2000TheCO}, the two-variable logic \cite{Lutz2001ModalLA} and the uniform one-dimensional fragment \cite{KuusistoUfSurvey}. They have also been considered in the context of description logics \cite{LutzS00}.

Boolean modal logics present novel and intriguing technical challenges, especially when studying the complexity of their satisfiability problems. Indeed, the most common explanation for the good algorithmic properties of modal logics is that they often enjoy some variant of the tree-model property \cite{Marx2007,WhyAreModalLogicsDecidable,Vardi1996WhyIM}, which Boolean modal logics lack. This is a consequence of the fact that Boolean modal logics have access to \emph{complements} of accessibility relations. To clarify this important point, we point out that in Boolean modal logics one can write statements such as $[R] [ \neg R ] \bot \land [\neg R] [\neg R ] \bot$, which enforce that $R$ needs to be interpreted in the Kripke models of this sentence as the total binary relation. Clearly such sentences do not have a tree-model property.

Satisfiability problems of Boolean modal logics were first studied in \cite{Lutz2000TheCO}, where it was proved that modal logic with negations of accessibility relations $\mathrm{ML}(\neg)$ has \textsc{ExpTime}-complete satisfiability problem, while the complexity of full Boolean modal logic is \textsc{NExpTime}-complete. As a follow-up to this work, in \cite{Lutz2001ModalLA} authors studied the complexity of the satisfiability problem for Boolean modal logic extended with inverses of accessibility relations and equality $\mathrm{ML}(I,s,\neg,\cap)$, and proved that if the number of accessibility relations is bounded, then its satisfiability problem is \textsc{ExpTime}-complete. This latter result should be contrasted with the fact that the satisfiability problem of $\mathrm{FO}^2$ is \textsc{NExpTime}-hard already over unary vocabularies, since $\mathrm{FO}^2$ has the same expressive power as $\mathrm{ML}(I,s,\neg,\cap)$ \cite{Lutz2001ModalLA}.

In this work we partially extend these results to the polyadic case. First, we will prove that the satisfiability problem of polyadic $\mathrm{ML}(\neg)$, which we denote by $\PML(\neg)$, has \textsc{ExpTime}-complete satisfiability problem. Secondly, we will partially generalize the complexity result of \cite{Lutz2001ModalLA} by proving that polyadic Boolean modal logic extended with arbitrary permutations of accessibility relations (as opposed to just inverses of binary accessibility relations), which we denote by $\PML(p,s,\neg,\cap)$, has \textsc{ExpTime}-complete satisfiability problem, if the number of accessibility relations in the underlying vocabulary is finite. We note that in general the satisfiability problem of $\PML(p,s,\neg,\cap)$ is \textsc{NExpTime}-complete, since it contains $\mathrm{ML}(\neg,\cap)$ and it is a fragment of $\mathrm{U}_1$ \cite{KuusistoUfSurvey}.

An important technical contribution of the present work is that we prove these two results in a \emph{unified manner}. Originally, the complexity of $\mathrm{ML}(\neg)$ was established using automata theory \cite{Lutz2000TheCO}, while the complexity of $\mathrm{ML}(I,s,\neg,\cap)$ over finite vocabularies was established via a poly-time reduction to the satisfiability problem of modal logic with \emph{difference operator} over a restricted class of Kripke frames. In contrast, we establish upper bounds on the satisfiability problem of $\PML(\neg)$ and $\PML(p,s,\neg,\cap)$ by simply using elementary (albeit in the second case quite technical) reductions to the satisfiability problem of polyadic $\mathrm{ML}$ extended with global diamond $\E$, which we denote by $\PML + \E$. These reductions were inspired by the reduction used in \cite{Lutz2001ModalLA} to establish the complexity of $\mathrm{ML}(I,s,\neg,\cap)$ over finite vocabularies. The overall structure of these reductions is quite robust and hence we expect that similar reductions will yield in the future several extensions of the results presented in this paper.

In addition to studying satisfiability problems, we also study the combined complexity of the model checking problem of $\PML(p,s,\neg,\cap)$, which is the variant where both the model and the formula itself are received as part of the input, and prove that it is \textsc{PTime}-complete, the non-trivial part being the upper bound. It is well-known that model checking problems of various modal logics with tree-model property -- even quite expressive ones such as the guarded fragment -- are \textsc{PTime}-complete \cite{ModelCheckingGF,GradelSurvey}. However, besides these modal logics and finite variable fragments of first-order logic \cite{GradelSurvey}, it seems that there are not that many natural logics known for which the complexity of the model checking problem lies in \textsc{PTime}. Note that the formulas of $\PML(p,s,\neg,\cap)$ are neither guarded nor are they expressible in any finite variable fragment of first-order logic. Thus, even though proving that the model checking problem of $\PML(p,s,\neg,\cap)$ is in \textsc{PTime} is not technically too complicated, we still believe that the result is interesting since the logic $\PML(p,s,\neg,\cap)$ is quite distinct from other logics known in the literature with \textsc{PTime}-complete combined complexity.

The structure of this paper is as follows. First, in Section \ref{sec:preliminaries} we will formally define the logics that are studied in this paper. Then, in Sections \ref{sec:negation_sat} and \ref{sec:bounded_sat} we prove that the satisfiability problems of $\PML(\neg)$ and $\PML(p,s,\neg,\cap)$ respectively are \textsc{ExpTime}-complete, the latter under the assumption that the number of accessibility relations is finite, which, as pointed out earlier, is necessary if $\textsc{ExpTime} \neq \textsc{NExpTime}$. Finally, in Section \ref{sec:model_checking} we prove that the combined complexity of $\PML(p,s,\neg,\cap)$ is \textsc{PTime}-complete.

\section{Preliminaries}\label{sec:preliminaries}

In this paper we will only consider relation symbols of arity at least two. Given a relation symbol $R$, we will use $ar(R)$ to denote its arity. If $\tau$ is a set of relation symbols and $\Phi$ is a set of propositional symbols, then a \emph{Kripke-model over} $(\tau,\Phi)$ is a tuple $\mathfrak{M} = (W,(R^\mathfrak{M})_{R\in \tau},V)$, where
\begin{enumerate}
    \item for every $R\in \tau$, $R^\mathfrak{M} \subseteq W^{ar(R)}$; and
    \item $V:\Phi \to \mathcal{P}(W)$.
\end{enumerate}
Members of $\tau$ are also called \emph{accessibility relations}. In what follows we will never specify explicitly what the underlying set of propositional symbols is and we will only occasionally specify the set of accessibility relations.

In this paper we consider extensions of standard polyadic (multimodal) modal logic via \emph{relation operators}, which are essentially isomorphic invariant mappings that map relational structures to relational structures. Since we will focus our attention only on a specific set of relation operators, we will omit the formal definition of a relation operator here, which can be found in \cite{Jaakkola2020AlgebraicCF}.

The main relation operators that we are going to consider are $p,s,\neg,\cap$, where $p$ and $s$ are called \emph{cyclic permutation} and \emph{swap permutation} respectively. Furthermore, for technical reasons we are going to need two additional relation operators $\backslash,\cup$. Let $\tau$ denote a set of relation symbols. Given $k\geq 2$, the set of $k$\emph{-ary terms} $\gra_k(p,s,\neg,\cap,\backslash,\cup)[\tau]$ is generated by the following grammar
\[\mathcal{R} ::= R \mid p\mathcal{R} \mid s\mathcal{R} \mid \neg \mathcal{R} \mid \mathcal{R} \cap \mathcal{R} \mid \mathcal{R} \backslash \mathcal{R} \mid \mathcal{R} \cup \mathcal{R},\]
where $R\in\tau$ is a $k$-ary relation. We use $\gra(p,s,\neg,\cap,\backslash,\cup)[\tau]$ to denote 
\[\bigcup_{k\geq 2} \gra_k(p,s,\neg,\cap,\backslash,\cup)[\tau],\]
which is the set of all terms. Arity of a term $\mathcal{R}$ is denoted by $ar(\mathcal{R})$.

Let $\mathfrak{M}$ be a Kripke model over $\tau$ and let $\mathcal{R} \in \gra(p,s,\neg,\cap,\backslash,\cup)$ be a $k$-ary term. We define the \emph{interpretation} of $\mathcal{R}$ over $\mathfrak{M}$ recursively as follows.
\begin{enumerate}
    \item If $\mathcal{R} = R \in \tau$, then we define $\llbracket R \rrbracket_\mathfrak{M} = R^\mathfrak{M}$.
    \item If $\mathcal{R} = p\mathcal{R}'$, then we define
    \[\llbracket \mathcal{R} \rrbracket_\mathfrak{M} = \{(a_k,a_1,\dots,a_{k-1}) \in W^k \mid (a_1,\dots,a_k) \in \llbracket \mathcal{R}' \rrbracket_\mathfrak{M}\}\]
    \item If $\mathcal{R} = s\mathcal{R}'$, then we define
    \[\llbracket \mathcal{R} \rrbracket_\mathfrak{M} = \{(a_1,\dots,a_{k-2},a_k,a_{k-1}) \in W^k \mid (a_1,\dots,a_k)\in \llbracket \mathcal{R}' \rrbracket_\mathfrak{M}\}\]
    \item If $\mathcal{R} = \neg \mathcal{R}'$, then we define $\llbracket \mathcal{R} \rrbracket_\mathfrak{M} = W^k \backslash \llbracket \mathcal{R}' \rrbracket_\mathfrak{M}$
    \item If $\mathcal{R} = \mathcal{R}' \cap \mathcal{R}''$, then we define $\llbracket \mathcal{R} \rrbracket_\mathfrak{M} = \llbracket \mathcal{R}' \rrbracket_\mathfrak{M} \cap \llbracket \mathcal{R}'' \rrbracket_\mathfrak{M}$
    \item If $\mathcal{R} = \mathcal{R}' \backslash \mathcal{R}''$, then we define $\llbracket \mathcal{R} \rrbracket_\mathfrak{M} = \llbracket \mathcal{R}' \rrbracket_\mathfrak{M} \backslash \llbracket \mathcal{R}'' \rrbracket_\mathfrak{M}$
    \item If $\mathcal{R} = \mathcal{R}' \cup \mathcal{R}''$, then we define $\llbracket \mathcal{R} \rrbracket_\mathfrak{M} = \llbracket \mathcal{R}' \rrbracket_\mathfrak{M} \cup \llbracket \mathcal{R}'' \rrbracket_\mathfrak{M}$
\end{enumerate}

Given $k\geq 2$, we let $S_k$ denote the set of all bijections $\{1,\dots,k\} \to \{1,\dots,k\}$. It is a well-known group theoretic fact that every permutation $\sigma \in S_k$ can be obtained by composing cyclic permutation with swap permutation. Given this, we will adopt the notational convention that we will use $\sigma \in S_k$ to denote some fixed (possibly empty) sequence consisting of operators $p$ and $s$ that generates it. The only explicit requirement that we impose on this sequence is that it should not be possible to rewrite it into a smaller one. Thus, for example, we use the identity permutation to denote the empty sequence.

The following lemma collects some elementary algebraic identities for terms.

\begin{lemma}\label{lemma:simple_equations}
    Let $\mathcal{R}_1,\mathcal{R}_2 \in \mathrm{GRA}(p,s,\neg,\cap)[\tau]$ be $k$-ary terms and let $\sigma \in S_k$. Let $\mathfrak{M}$ be a Kripke model over $\tau$.
    \begin{enumerate}
        \item $\llbracket \neg \neg \mathcal{R}_1 \rrbracket_\mathfrak{M} = \llbracket \mathcal{R}_1 \rrbracket_\mathfrak{M}$
        \item $\llbracket \sigma \neg \mathcal{R}_1 \rrbracket_\mathfrak{M} = \llbracket \neg \sigma \mathcal{R}_1 \rrbracket_\mathfrak{M}$
        \item $\llbracket \neg (\mathcal{R}_1 \cap \mathcal{R}_2) \rrbracket_\mathfrak{M} = \llbracket (\neg \mathcal{R}_1 \cup \neg \mathcal{R}_2) \rrbracket_\mathfrak{M}$
        \item $\llbracket \neg (\mathcal{R}_1 \cup \mathcal{R}_2) \rrbracket_\mathfrak{M} = \llbracket (\neg \mathcal{R}_1 \cap \neg \mathcal{R}_2) \rrbracket_\mathfrak{M}$
        \item $\llbracket (\mathcal{R}_1 \cap \neg \mathcal{R}_2) \rrbracket_\mathfrak{M} = \llbracket (\mathcal{R}_1 \backslash \mathcal{R}_2) \rrbracket_\mathfrak{M}$
        \item $\llbracket (\mathcal{R}_1 \cup \neg \mathcal{R}_2) \rrbracket_\mathfrak{M} = \llbracket \neg (\mathcal{R}_2 \backslash \mathcal{R}_1) \rrbracket_\mathfrak{M}$
    \end{enumerate}
\end{lemma}

Let $\mathcal{R} \in \mathrm{GRA}(p,s,\neg)[\tau]$ be a $k$-ary term. We say that $\mathcal{R}$ is a $k$-\emph{literal over} $\tau$, if it is either of the form $\neg \sigma R$ or $\sigma R$, for a $k$-ary relation symbol $R\in \tau$. A maximally consistent set $\rho$ of $k$-literals over $\sigma$ is called a $k$-\emph{table} over $\tau$. We identify tables $\rho$ with terms $\bigcap_{\alpha \in \rho} \alpha$.

Given a Kripke model $\mathfrak{M}$ over $\tau$ and $(w_1,\dots,w_k) \in W^k$ we say that $(w_1,\dots,w_k)$ \emph{realizes} a $k$-table $\rho$, if for every $k$-literal $\alpha$ we have that
\[(w_1,\dots,w_k) \in \llbracket \alpha \rrbracket_\mathfrak{M} \Leftrightarrow \alpha \in \rho.\]
Note that since tables are maximally consistent, each tuple in a given Kripke model realizes a unique table. Conversely, the accessibility relations of a Kripke model can be described completely by specifying what tables different tuples realize.

Given a $k$-table $\rho$ and $\sigma \in S_k$ we define
\[\sigma[\rho] := \bigcap_{\sigma' R \in \rho} (\sigma \circ \sigma')R \cap \bigcap_{\neg \sigma' R\in \rho} \neg (\sigma \circ \sigma')R.\]
Note that for every $k$-table $\rho$ and $\sigma_1,\sigma_2 \in S_k$ we have that $\sigma_1[\sigma_2[\rho]] = (\sigma_1 \circ \sigma_2)[\rho]$. Furthermore, we note that it can be the case that $\sigma[\rho] = \rho$, even if $\sigma$ is not the identity permutation.\footnote{Consider for example the permutation $s$ and the $2$-table $\{R,sR\}$ over $\{R\}$, where $R$ is a binary relation.}

Let $\Phi$ be a set of propositional symbols, $\tau$ a set of relation symbols and $\mathcal{F} \subseteq \{p,s,\neg,\cap,\backslash,\cup\}$. The set of formulas $\mathrm{PML}(\mathcal{F})[\tau,\Phi]$ is generated by the following grammar
\[\varphi ::= p \mid \neg \varphi \mid (\varphi \land \varphi) \mid \langle \mathcal{R} \rangle (\underbrace{\varphi,\dots,\varphi}_{k\text{-times}})\]
where $p \in \Phi$ and $\mathcal{R} \in \gra(\mathcal{F})[\tau]$ is a $(k+1)$-ary term. We will use standard shorthand notations such as $\varphi \lor \psi := \neg (\neg \varphi \land \neg \psi)$ and $[\mathcal{R}] := \neg \langle \mathcal{R} \rangle \neg$.  Given a formula $\varphi$ we let $\subf(\varphi)$ denote the set of subformulas of $\varphi$.

We will use $\PML(\mathcal{F})[\tau,\Phi] + \E$ to denote the set of formulas generated by the grammar of $\PML(\mathcal{F})[\tau,\Phi]$ extended with the rule
\[\varphi ::= \E \varphi\]
If $\tau$ contains only binary relation symbols, then we emphasize this by writing $\mathrm{ML}(\mathcal{F})[\tau,\Phi]$. If we are only considering sets of relation symbols $\tau$ of size at most some fixed constant $c$, then we emphasize this by writing $\PML_c[\Phi]$. Thus $\PML_c[\tau,\Phi]$ entails that $|\tau| \leq c$. As in the case of Kripke models, for the rest of this paper we will never explicitly mention the underlying set of propositional symbols.

The semantics of our logics are standard and we will only recall here the semantic clauses of $\langle \mathcal{R} \rangle (\psi_1,\dots,\psi_k)$ and $\E \psi$. Given a Kripke model $\mathfrak{M}$ over $\tau$ and $w\in W$ we define
\[\mathfrak{M},w \Vdash \langle \mathcal{R} \rangle (\psi_1,\dots,\psi_k) \Leftrightarrow\]
\[\text{there exists } (w,w_1,\dots,w_k) \in \llbracket \mathcal{R} \rrbracket_\mathfrak{M} \text{ such that } \mathfrak{M},w_\ell \Vdash \psi_\ell \text{, for every } 1\leq \ell \leq k\]
and
\[\mathfrak{M},w \Vdash \E \psi \Leftrightarrow \text{there exists some } w' \in W \text{ such that } \mathfrak{M},w' \Vdash \psi.\]
We will conclude this section by stating the complexity of the satisfiability problem of $\PML + \E$.

\begin{proposition}
    The satisfiability problem of $\PML + \E$ is \textsc{ExpTime}-complete.
\end{proposition}
\begin{proof}
    Lower bound follows from the well-known fact that the satisfiability problem of $\mathrm{ML} + \E$ is \textsc{ExpTime}-complete \cite[Exercise 6.8.1]{DBLP:books/cu/BlackburnRV01}. Upper bound follows, say, from the recent result that the satisfiability problem of the so-called \emph{guarded forward fragment} FGF is \textsc{ExpTime}-complete \cite{Bednarczyk2021ExploitingFS}. Indeed, by using variables carefully, one can guarantee that the standard translation of $\PML + \E$ into FO produces sentences of FGF.
\end{proof}

As mentioned in the introduction, existing results in the literature imply that the satisfiability problem of $\PML(p,s,\neg,\cap)$ is \textsc{NExpTime}-complete. 

\begin{proposition}
    The satisfiability problem of $\PML(p,s,\neg,\cap)$ is \textsc{NExpTime}-complete.
\end{proposition}
\begin{proof}
    Lower bound follows from the fact that the satisfiability problem of $\mathrm{ML}(\neg,\cap)$ is already \textsc{NExpTime}-hard \cite{ComplexityReasoningBooleanModalLogics}. For upper bound we note that $\PML(p,s,\neg,\cap)$ is contained in $\mathrm{U}_1$ \cite[Theorem 7]{KuusistoUfSurvey}, for which the satisfiability problem is \textsc{NExpTime}-complete \cite{Kieronski2014ComplexityAE}.
\end{proof}

\section{Satisfiability problem of $\PML(\neg)$}\label{sec:negation_sat}

In this section we will show how the satisfiability problem of $\PML(\neg)$ can be reduced efficiently (in polynomial time) to that of $\PML + \E$, which will yield the following result.

\begin{theorem}\label{thm:neg_sat_complexity}
    The satisfiability problem of $\PML(\neg)$ is \textsc{ExpTime}-complete.
\end{theorem}
It seems most probable that the above complexity result continues to hold for $\PML(\sigma,\neg)$, but we have not yet been able to show this.

Before presenting the reduction from $\PML(\neg)$ to $\PML + \E$, we first describe one brief application of Theorem \ref{thm:neg_sat_complexity}, which reflects one original motivation for the study of Boolean modal logics \cite{ComplexityReasoningBooleanModalLogics}. Namely, we show that the logic $\PML$ extended with \emph{polyadic window operator}, which we denote by $\PML + \nabla$, has a \textsc{ExpTime}-complete satisfiability problem. Consider a vocabulary $\tau$. The grammar for generating the formulas of $\PML[\tau] + \nabla$ is the grammar of $\PML[\tau]$ extended with the rule
\[\varphi ::= \nabla_R (\underbrace{\varphi,\dots,\varphi}_{k\text{-times}}),\]
where $R$ is a $(k+1)$-ary accessibility relation, for every $R\in \tau$. The semantics of formulas of the form $\nabla_R(\psi_1,\dots,\psi_k)$ is defined as follows:
\[\mathfrak{M},w \Vdash \nabla_R (\psi_1,\dots,\psi_k) \Leftrightarrow \text{For every } (w_1,\dots,w_k) \in W^k \text{ we have that if}\]
\[\mathfrak{M},w_\ell \Vdash \psi_\ell \text{, for every $1\leq \ell \leq k$, then } (w,w_1,\dots,w_k) \in R^\mathfrak{M}.\]
It is easy to see that $\mathfrak{M},w \Vdash \nabla_R(\psi_1,\dots,\psi_k)$ is equivalent with
\[\mathfrak{M},w \Vdash [ \neg R ] (\neg \psi_1,\dots,\neg \psi_k)\]
and hence Theorem \ref{thm:neg_sat_complexity} immediately implies the following complexity results.

\begin{theorem}
    The satisfiability problem of $\PML + \nabla$ is \textsc{ExpTime}-complete.
\end{theorem}

Now we will present the reduction. Fix a formula $\varphi \in \PML(\neg)$ and let $\tau$ denote the set of accessibility relations occuring in $\varphi$. For every symbol $R \in \tau$, we will introduce two fresh symbols of the same arity, $R_1$ and $R_2$. Given $\psi \in \mathrm{Subf}(\varphi)$, we let $t(\psi)$ denote the formula obtained from $\psi$ by replacing each $\langle R \rangle$ with $\langle R_1 \rangle$ and each $\langle \neg R \rangle$ with $\langle R_2 \rangle$. Consider then the following formula $\theta := t(\varphi) \land \eta$, where
\[\eta := \bigwedge_{\substack{\langle R_1 \rangle (\psi_1,...,\psi_k), \\ \langle R_2 \rangle (\chi_1,...,\chi_k) \\ \in \subf(t(\varphi))}} \bigg(\E (\neg \langle R_1 \rangle (\psi_1,\dots,\psi_k) \land \neg \langle R_2 \rangle (\chi_1,\dots,\chi_k))\]
\[\to \bigvee_{1\leq \ell \leq k} \A (\neg \psi_\ell \lor \neg \chi_\ell) \bigg).\]
Here (and elsewhere) $\A$ is shorthand notation for $\neg \E \neg$. Intuitively speaking, in every model of $\eta$ we can extend the interpretations of $R_1$ and $R_2$, for $R \in \tau$, in such a way that they cover $W^{ar(R)}$, i.e., every tuple of length $ar(R)$ belongs either to the interpretation of $R_1$ or to the interpretation of $R_2$, while maintaining that the resulting model is a model of $t(\varphi)$, if the original model was.

Since the big conjunction in $\eta$ ranges over only those formulas that occur as subformulas in $\varphi$, the size of $\eta$ is $O(|\varphi|^2)$, i.e., polynomial with respect to $|\varphi|$. The rest of this section is devoted to proving that the above reduction is correct, i.e., $\varphi$ is satisfiable iff $\theta$ is. We will start with the left to right direction.

\begin{lemma}
    If $\varphi$ is satisfiable, then so is $\theta$.
\end{lemma}
\begin{proof}
    Let $\mathfrak{M} = (W,(R)_{R\in \tau},V)$ be a Kripke model and let $w\in W$ be a world so that $\mathfrak{M},w \Vdash \varphi$. We then define the Kripke model $\mathfrak{N} = (W,(R_1)_{R\in \tau},(R_2)_{R\in \tau}, V)$ by setting that for every $R\in \tau$, $R_1^\mathfrak{N} = R^\mathfrak{M}$ and $R_2^\mathfrak{N} = W^{ar(R)} \backslash R^\mathfrak{M}$. Clearly $\mathfrak{N},w \Vdash t(\varphi)$. To verify that $\mathfrak{N}$ satisfies $\eta$, suppose that there exists $\langle R \rangle (\psi,...,\psi_k),\langle \neg R \rangle (\chi_1,...,\chi_k) \in \mathrm{Subf}(t(\varphi))$ and $w_0 \in W$ so that
    \[\mathfrak{N},w_0 \Vdash \neg \langle R_1\rangle (\psi_1,\dots,\psi_k) \land \neg \langle R_2 \rangle (\chi_1,\dots,\chi_k)\]
    but 
    \[\mathfrak{N},w \not\Vdash \bigvee_{1\leq \ell \leq k} \A(\neg \psi_\ell \lor \neg \chi_\ell).\]
    Thus for every $1\leq \ell \leq k$ there exists $w_\ell \in W$ so that $\mathfrak{N},w_\ell \Vdash \psi_\ell \land \chi_\ell$. By construction, we must either have that $(w_0,w_1,...,w_k) \in R_1^\mathfrak{N}$ or $(w_0,w_1,...,w_k)\in R_2^\mathfrak{N}$, but clearly both of these cases lead to a contradiction.
\end{proof}

Suppose then that $\theta$ is satisfiable. $\mathfrak{M}$ be a Kripke model and let $w \in W$ be a world so that $\mathfrak{M},w \Vdash \theta$.

\begin{lemma}\label{lemma:negation_covering_property}
    For every $(w_0,w_1,\dots,w_k)\in W^{k+1}$ and $R \in \tau$ there exists $i\in \{1,2\}$ so that for every $\langle R_i \rangle (\psi_1,\dots,\psi_k) \in \subf(t(\varphi))$ we have that if $\mathfrak{M},w_\ell \Vdash \psi_\ell$, for every $1\leq \ell \leq k$, then $\mathfrak{M},w_0 \Vdash \langle R_i \rangle (\psi_1,\dots,\psi_k)$.
\end{lemma}
\begin{proof}
    Suppose that this is not the case. Thus there exists
    \[\langle R_1 \rangle (\psi_1,\dots,\psi_k), \langle R_2 \rangle (\chi_1,\dots,\chi_k) \in \subf(t(\varphi))\]
    so that $\mathfrak{M},w_\ell \Vdash \psi_\ell \land \chi_\ell$, for every $1\leq \ell \leq k$, but
    \[\mathfrak{M},w_0 \Vdash \neg \langle R_1 \rangle (\psi_1,\dots,\psi_k) \land \neg \langle R_2 \rangle (\chi_1,\dots,\chi_k).\]
    Since $\mathfrak{M},w \Vdash \eta$, we have that
    \[\mathfrak{M},w \Vdash \bigvee_{1\leq \ell \leq k} \A (\neg \psi_\ell \lor \neg \chi_\ell),\]
    which is a clear contradiction.
\end{proof}

Using Lemma \ref{lemma:negation_covering_property}, we can extend the model $\mathfrak{M}$ as follows. For every $(w_1,\dots,w_k) \not\in R_1^\mathfrak{M} \cup R_2^\mathfrak{M}$, we choose $i\in \{1,2\}$ with the properties described in Lemma \ref{lemma:negation_covering_property}, and add $(w_1,\dots,w_k)$ to $R_i^\mathfrak{M}$. We still use $\mathfrak{M}$ to denote the resulting model. We emphasize that $\mathfrak{M}$ has now the property that for every $(w_1,\dots,w_k)\in W^k$ and $R\in \tau$, either $(w_1,\dots,w_k) \in R_1^\mathfrak{M}$ or $(w_1,\dots,w_k) \in R_2^\mathfrak{M}$.

\begin{lemma}
    $\mathfrak{M},w \Vdash t(\varphi)$
\end{lemma}
\begin{proof}
    A routine induction.
\end{proof}

We now define a Kripke model $\mathfrak{N} = (W^*, (R^\mathfrak{N})_{R\in \tau}, V^*)$ over $\tau$ as follows. First, we specify that $W^* := W \times \{0,1\}$ and that for every $(w,i) \in W^*$ we have that $(w,i) \in V^*(p)$ iff $w\in V(p)$. Next we need to define interpretations of relation symbols $R\in \tau$. Fix such a relation symbol $R$. For every $(w_0,i) \in W^*$ we define that if $(w_0,w_1,\dots,w_k) \in R_1^\mathfrak{M}$, then 
\[((w_0,i),(w_1,i+1 \mod 2), \dots, (w_k,i+1 \mod 2)) \in R^\mathfrak{N}.\]
Then, for every $(w_0,w_1,\dots,w_k) \in R_2^\mathfrak{M}$ we define that
\[((w_0,i),(w_1,i),\dots,(w_k,i)) \not\in R^\mathfrak{N}.\]
Finally, for every $((w_0,i_0),\dots,(w_k,i_k))$ for which we have not specified whether they belong to $R^\mathfrak{N}$, we define that if $(w_0,\dots,w_k) \not\in R_2^\mathfrak{M}$ then $((w_0,i_0),\dots,(w_k,i_k)) \in R^\mathfrak{N}$.

\begin{lemma}\label{lemma:negation_final_model_construction}
    For every $\psi \in \subf(\varphi)$ and $w_0\in W$ we have that
    \[\mathfrak{N},(w_0,i_0) \Vdash \psi \Leftrightarrow \mathfrak{M},w_0 \Vdash t(\psi).\]
\end{lemma}
\begin{proof}
    A routine induction.
\end{proof}

In particular $\mathfrak{N}$ is a model of $\varphi$, since $\mathfrak{N},(w,0) \Vdash \varphi$, and hence $\varphi$ is satisfiable. Thus we can conclude that $\varphi$ is satisfiable iff $\theta$ is.

\section{Satisfiability problem of $\PML_c(p,s,\neg,\cap)$}\label{sec:bounded_sat}

In this section we will reduce the satisfiability problem of $\PML_c(p,s,\neg,\cap)$ -- the restriction of $\PML(p,s,\neg,\cap)$ where the underlying vocabulary can contain at most $c$ accessibility relations, where $c$ is a fixed constant -- to that of $\PML + \E$, which will yield the following result.

\begin{theorem}\label{thm:finite_sat_complexity}
    The satisfiability problem of $\PML_c(p,s,\neg,\cap)$ is \textsc{ExpTime}-complete.
\end{theorem}

The reduction used in the proof of Theorem \ref{thm:finite_sat_complexity} is very similar to the reduction that was used in the previous section to prove Theorem \ref{thm:neg_sat_complexity}. The reader is encouraged to keep this in mind when parsing the reduction, since even though the underlying ideas are again elementary, the resulting reduction is quite technical.

An important property of $\PML(\neg)$ is that it can not speak about intersections of accessibility  relations. In the case of $\PML(p,s,\neg,\cap)$ this is obviously no longer the case, but it is still possible to convert each sentence of $\PML(p,s,\neg,\cap)$ into an equi-satisfiable sentence with an analogous property. If the number of underlying accessibility relations is bounded by some constant, then this translation can also be carried out in polynomial time.

\begin{lemma}\label{lemma:accessability_to_table}
    Let $\varphi \in \PML_c(p,s,\neg,\cap)[\tau]$ be a formula. Then we can transform $\varphi$ in polynomial time to a formula $\varphi^* \in \PML_c(p,s,\neg,\cap)[\tau]$, which has the following properties.
    \begin{enumerate}
        \item $\varphi$ is satisfiable if and only if $\varphi^*$ is.
        \item For every $\langle \mathcal{R} \rangle (\psi_1,\dots,\psi_k) \in \subf(\varphi^*)$ the term $\mathcal{R}$ is a $k$-table over $\tau$.
    \end{enumerate}
\end{lemma}
\begin{proof}
    By applying repeatedly Lemma \ref{lemma:simple_equations}, we can assume that for every $\langle \mathcal{R} \rangle (\psi_1,\dots,\psi_k) \in \subf(\varphi)$ we have that $\mathcal{R}$ is a boolean combination of $(k+1)$-literals over $\tau$. Now we pick an innermost such subformula of $\varphi$. The term $\mathcal{R}$ is clearly equivalent with the term
    \[\bigcup_{\rho \models \mathcal{R}} \ \bigcap_{\alpha \in \rho} \ \alpha,\]
    where each $\rho$ is a $(k+1)$-table over $\tau$. Let $p_{\psi_1},\dots,p_{\psi_k}$ denote fresh propositional symbols. In $\varphi$ we replace $\langle R \rangle (\psi_1,\dots,\psi_k)$ with the following formula
    \[\bigvee_{\rho \models \mathcal{R}} \langle \bigcap_{\alpha \in \rho} \alpha \rangle (p_{\psi_1},\dots,p_{\psi_k}).\]
    Let $\varphi'$ denote the resulting formula. Without loss of generality we will assume that $\tau$ contains at least one binary relation symbol. With this technical assumption it is clear that $\varphi$ is equisatisfiable with the formula
    \[\varphi' \land \bigwedge_{2\text{-table } \rho} \ [ \bigcap_{\alpha \in \rho} \alpha](p_{\psi_\ell} \leftrightarrow \psi_\ell).\]
    Since the size of $\tau$ is bounded by a constant, the size of the above formula is polynomial with respect to the size of the original formula $\varphi$. (Indeed, the fact that $|\tau|$ is bounded by a constant entails that the number of tables over $\tau$ is also bounded by a constant.) By repeating the above procedure sufficiently many times, we will eventually reach the desired formula.
\end{proof}

\begin{remark}
    In Lemma \ref{lemma:accessability_to_table} it is not enough to assume that there is a constant bound on the \emph{arities} of relations in $\tau$ to guarantee that the translation is efficient, i.e., polynomial time computable, since even the number of $2$-tables over a binary vocabulary $\tau$ is bounded from below by $2^{|\tau|}$.
\end{remark}

For the rest of this section we will assume that $\varphi$ satisfies property (ii) of Lemma \ref{lemma:accessability_to_table}. Now, for $2\leq k\leq m$, where $m = \max \{ar(R) \mid R\in \tau\}$, let $\mathfrak{T}_k$ denote the set of all tables over $\tau$. Each $k$-table $\rho \in \mathfrak{F}_k$ can be seen as a $k$-ary accessibility  relation and hence we will consider the vocabulary $\tau_\mathfrak{T} := \bigcup_{2\leq k\leq m} \mathfrak{T}_k$. We let $t(\varphi)$ denote the sentence in $\PML_c[\tau_\mathfrak{T}]$ which is obtained from $\varphi$ by replacing each table $\rho \in \gra(p,s,\neg,\cap)[\tau]$ with the corresponding relation symbol $\rho \in \tau_\mathfrak{T}$.

We next describe sentences of $\PML[\tau_\mathfrak{T}] + \E$ which together play the same role that $\eta$ did in the previous section. We start with the following sentence, which we denoted by $\xi_1$.
\[\xi_1 := \A \ \bigwedge_{1\leq k < m} \ \bigwedge_{\substack{g:\mathfrak{T}_{k+1} \to S_{k+1} \\ \sigma_\rho := g(\rho)}} \ \bigwedge_{\substack{h:\mathfrak{T}_{k+1} \to (\subf(t(\varphi)))^k \\ h(\rho) = (\psi_1^{\sigma_\rho},\dots,\psi_k^{\sigma_\rho})}} \Bigg( \bigg ( \bigwedge_{\rho \in X_0} \neg \langle \sigma_\rho[\rho] \rangle (\psi_1^{\sigma_\rho},\dots,\psi_k^{\sigma_\rho})\]
\[\land \bigwedge_{1\leq \ell \leq k} \E \bigg(\bigwedge_{\rho \in X_\ell} \neg \langle \sigma_\rho [\rho] \rangle (\psi_1^{\sigma_\rho}, \dots, \psi_k^{\sigma_\rho}) \land \bigwedge_{\rho \not\in X_\ell} \psi_{\sigma_\rho^{-1}(\ell)}^{\sigma_\rho} \bigg) \bigg) \to \bigvee_{\rho \not\in X_0} \neg \psi_{\sigma_\rho^{-1}(0)}^{\sigma_\rho} \Bigg)\]
In the above sentence we use $X_\ell$ to denote the set $\{\rho \mid \sigma_\rho(0) = \ell\}$. Note that since the size of $\tau$ is bounded by a constant, $\xi_1$ is only of size at most polynomial with respect to the size of $\varphi$. In addition to $\xi_1$, we will need the following sentences, which will be denoted by $\xi_2$ and $\xi_3$ respectively.
\[\xi_2 := \A \ \bigwedge_{1\leq k < m} \ \bigwedge_{\rho \in \mathfrak{T}_{k+1}} \ \bigwedge_{\substack{\sigma \in S_{k+1} \\ \sigma(0) \neq 0}} \ \bigwedge_{\psi_1,\dots,\psi_k \in \subf(t(\varphi))}\]
\[\bigg(\neg \psi_{\sigma^{-1}(0)} \lor \neg \langle \rho \rangle (\neg \psi_{\sigma^{-1}(1)},\dots,\underbrace{\langle \sigma[\rho] \rangle(\psi_1,\dots,\psi_k)}_{\sigma(0)\text{:th formula}},\dots,\neg \psi_{\sigma^{-1}(k)}) \bigg)\]
\[\xi_3 := \A \ \bigwedge_{1\leq k < m} \ \bigwedge_{\rho \in \mathfrak{T}_{k+1}} \ \bigwedge_{\substack{\sigma \in S_{k+1} \\ \sigma(0) = 0}} \ \bigwedge_{\psi_1,\dots,\psi_k \in \subf(t(\varphi))}\]
\[\bigg(\langle \rho \rangle (\psi_{\sigma^{-1}(1)},\dots,\psi_{\sigma^{-1}(k)}) \to \langle \sigma[\rho] \rangle (\psi_1,\dots,\psi_k) \bigg)\]
We let $\Theta := t(\varphi) \land \xi_2 \land \xi_2 \land \xi_3$. The sentences $\xi_1,\xi_2,\xi_3$ might look rather complicated, but we emphasize again that they are simply playing essentially the same role that $\eta$ did in the previous section. Namely, they axiomatize enough properties of tables so that in any model of $\Theta$ we can enlarge  the interpretations of the accessibility relations $\rho$ in such a way that they cover all the tuples of relevant length, while maintaining that the resulting model is still a model of $t(\varphi)$.

The rest of this section is devoted towards proving that $\varphi$ is satisfiable iff $\Theta$. We start with the easy direction.

\begin{lemma}\label{lemma:pmlm_easy_lemma}
    If $\varphi$ is satisfiable, then so is $\Theta$. 
\end{lemma}
\begin{proof}
    Let $\mathfrak{M} = (W,(R^\mathfrak{M})_{R\in \tau},V)$ be a Kripke model and let $w\in W$ be a world so that $\mathfrak{M},w \Vdash \varphi$. We then define the Kripke model $\mathfrak{N} = (W,(\rho^\mathfrak{N})_{\rho \in \tau_\mathfrak{T}}, V)$ by setting that for every $\rho\in \tau_\mathfrak{T}$, $\rho^\mathfrak{N} = \llbracket \rho \rrbracket_\mathfrak{M}$. Clearly $\mathfrak{N},w \Vdash t(\varphi)$.
    
    Let us then verify that $\mathfrak{N},w_0 \Vdash \xi_1$, for any $w_0 \in W$. Fix $k, g$ and $h$. Suppose that
    \[\mathfrak{N},w_0 \Vdash \bigwedge_{\rho \in X_0} \neg \langle \sigma_\rho[\rho] \rangle (\psi_1^{\sigma_\rho},\dots,\psi_k^{\sigma_\rho})\]
    and that for every $1\leq \ell \leq k$ there exists $w_\ell$ so that
    \[\mathfrak{N},w_\ell \Vdash \bigwedge_{\rho \in X_\ell} \neg \langle \sigma_\rho [\rho] \rangle (\psi_1^{\sigma_\rho}, \dots, \psi_k^{\sigma_\rho}) \land \bigwedge_{\rho \not\in X_\ell} \psi_{\sigma_\rho^{-1}(\ell)}^{\sigma_\rho}.\]
    Our goal is to show that
    \[\mathfrak{N},w_0 \Vdash \bigvee_{\rho \not\in X_0} \neg \psi_{\sigma_\rho^{-1}(0)}^{\sigma_\rho}.\]
    Aiming for a contradiction, suppose that this is not the case. Since every tuple realizes a table in $\mathfrak{M}$, there exists $\rho \in \tau_\mathfrak{T}$ so that $(w_0,w_1,\dots,w_k) \in \rho^\mathfrak{N}$. Recall that the function $g$ associates a permutation $\sigma_\rho$ with $\rho$. Now $(w_{\sigma_\rho(0)},\dots,w_{\sigma_\rho(k)}) \in \llbracket \sigma_\rho \rho \rrbracket_\mathfrak{M}$. Let $\ell_0 := \sigma_\rho(0)$. Then, for every $\ell \neq \ell_0$ we have that
    \[\mathfrak{N},w_\ell \Vdash \psi_{\sigma_\rho^{-1}(\ell)}^{\sigma_\rho},\]
    since $\rho \not\in X_\ell$. This implies that
    \[\mathfrak{N},w_{\sigma_\rho(\ell)} \Vdash \psi_\ell^{\sigma_\rho},\]
    for every $1\leq \ell \leq k$. On the other hand
    \[\mathfrak{N},w_{\sigma_\rho(0)} \Vdash \neg \langle \sigma_\rho[\rho] \rangle (\psi_1^{\sigma_\rho},\dots,\psi_k^{\sigma_\rho}),\]
    which is a contradiction, since $\sigma_\rho[\rho]^\mathfrak{N} = \llbracket \sigma_\rho \rho \rrbracket_\mathfrak{M}$.
    
    Next, we will verify that $\mathfrak{N},w_0 \Vdash \xi_2$, for any $w_0 \in W$. Fix $k, \sigma$ and $\psi_1,\dots,\psi_k \in \subf(t(\varphi))$. Aiming for a contradiction, suppose that
    \[\mathfrak{N},w_0 \Vdash \bigg(\psi_{\sigma^{-1}(0)} \land \langle \rho \rangle ( \psi_{\sigma^{-1}(1)},\dots,\neg \langle \sigma[\rho](\psi_1,\dots,\psi_k),\dots, \psi_{\sigma^{-1}(k)}) \bigg)\]
    Thus there exists $(w_0,w_1,\dots,w_k) \in \rho^\mathfrak{N}$ so that $\mathfrak{M},w_\ell \Vdash \psi_{\sigma^{-1}(\ell)}$, for every $\ell \neq \sigma(0)$ (including $\ell = 0$), and
    \[\mathfrak{M},w_{\sigma(0)} \Vdash \neg \langle \sigma[\rho] \rangle (\psi_1,\dots,\psi_k).\]
    Now $(w_{\sigma(0)},\dots,w_{\sigma(k)}) \in (\sigma[\rho])^\mathfrak{M}$ and furthermore $\mathfrak{M},w_{\sigma(\ell)} \Vdash \psi_\ell$, for every $1\leq \ell \leq k$, which is a contradiction.
    
    Finally, we will verify that $\mathfrak{N},w_0 \Vdash \xi_3$, for any $w_0 \in W$. Fix $k,\sigma$ and $\psi_1,\dots,\psi_k \in \subf(t(\varphi))$. Aiming for a contradiction, suppose that
    \[\mathfrak{M},w_0 \Vdash \langle \rho \rangle (\psi_{\sigma^{-1}(1)},\dots,\psi_{\sigma^{-1}(k)}),\]
    but
    \[\mathfrak{M},w_0 \Vdash \neg \langle \sigma [\rho] \rangle (\psi_1,\dots,\psi_k).\]
    Now there exists $(w_0,w_1,\dots,w_k) \in \rho^\mathfrak{N}$ so that $\mathfrak{M},w_\ell \Vdash \psi_{\sigma^{-1}(\ell)}$. Hence $(w_0,w_{\sigma(1)},\dots,w_{\sigma(k)}) \in \llbracket \sigma \rho \rrbracket_\mathfrak{M} = (\sigma[\rho])^\mathfrak{N}$, which is a contradiction, since $\mathfrak{M},w_{\sigma(\ell)} \Vdash \psi_\ell$, for every $1\leq \ell \leq k$.
\end{proof}

Suppose now that $\mathfrak{M}$ is a Kripke model and that there is a $w\in W$ so that $\mathfrak{M},w \Vdash \Theta$. Our goal is to use $\mathfrak{M}$ to construct a model for $\varphi$. First, we will need a lemma which guarantees that we can extend the interpretations of $k$-ary accessibility relations in such a way that every tuple will belong to the interpretation of at least one such relation.

\begin{lemma}\label{lemma:pmlm_main_lemma}
    For every $1\leq k < m$ and $(w_0,w_1,\dots,w_k) \in W^{k+1}$ there exists $\rho \in \mathfrak{T}_{k+1}$ so that for all $\sigma \in S_{k+1}$ and $\psi_1,\dots,\psi_k \in \subf(t(\varphi))$ we have that if $\mathfrak{M},w_{\sigma(\ell)} \Vdash \psi_\ell$, for every $1\leq \ell \leq k$, then $\mathfrak{M},w_{\sigma(0)} \Vdash \langle \sigma[\rho] \rangle (\psi_1,\dots,\psi_k)$.
\end{lemma}
\begin{proof}
    Fix $(w_0,w_1,\dots,w_k) \in W^{k+1}$. Aiming for a contradiction, suppose that for every $\rho \in \mathfrak{T}_{k+1}$ there exists $\sigma_\rho \in S_{k+1}$ and $\psi_1^{\sigma_\rho},\dots,\psi_k^{\sigma_\rho} \in \subf(t(\varphi))$ so that
    \[\mathfrak{M},w_{\sigma_\rho(\ell)} \Vdash \psi_\ell^{\sigma_\rho},\]
    for every $1\leq \ell \leq k$, but
    \[\mathfrak{M},w_{\sigma_\rho(0)} \Vdash \neg \langle \sigma_\rho[\rho] \rangle (\psi_1^{\sigma_\rho},\dots,\psi_k^{\sigma_\rho}).\]
    It is simple to verify that this entails that
    \[\mathfrak{M},w_0 \Vdash \bigwedge_{\rho \in X_0} \neg \langle \sigma_\rho[\rho] \rangle (\psi_1^{\sigma_\rho},\dots,\psi_k^{\sigma_\rho}) \]
    \[\land \bigwedge_{1\leq \ell \leq k} \E \bigg(\bigwedge_{\rho \in X_\ell} \neg \langle \sigma_\rho [\rho] \rangle (\psi_1^{\sigma_\rho}, \dots, \psi_k^{\sigma_\rho}) \land \bigwedge_{\rho \not\in X_\ell} \psi_{\sigma_\rho^{-1}(\ell)}^{\sigma_\rho} \bigg)\]
    and since $\mathfrak{M},w_0 \Vdash \xi_1$, we have that
    \[\mathfrak{M},w_0 \Vdash \bigvee_{\rho \not\in X_0} \neg \psi_{\sigma_\rho^{-1}(0)}^{\sigma_\rho},\]
    which is a contradiction, since by assumption $\mathfrak{M},w_0 \Vdash \bigwedge_{\rho \not\in X_0} \psi_{\sigma_\rho^{-1}(0)}^{\sigma_\rho}$.
\end{proof}

Secondly, we will need a lemma which guarantees that we can close the interpretations of accessibility relations under permutations.

\begin{lemma}\label{lemma:pmlm_closed_under_permutations}
    For every $1\leq k < m$, $\rho \in \mathfrak{T}_{k+1}$, $\sigma \in S_{k+1}$ we have that if $(w_0,w_1,\dots,w_k) \in \rho^\mathfrak{M}$, then for every $\psi_1,\dots,\psi_k \in \subf(t(\varphi))$ we have that if $\mathfrak{M},w_{\sigma(\ell)} \Vdash \psi_\ell$, for every $1\leq \ell \leq k$, then $\mathfrak{M},w_{\sigma(0)} \Vdash \langle \sigma[\rho] \rangle (\psi_1,\dots,\psi_k)$.
\end{lemma}
\begin{proof}
    Fix $k,\rho,\sigma$ and $(w_0,w_1,\dots,w_k) \in \rho^\mathfrak{M}$. Aiming for a contradiction, suppose that there exists $\psi_1,\dots,\psi_k \in \subf(t(\varphi))$ so that
    \[\mathfrak{M},w_{\sigma(\ell)} \Vdash \psi_\ell,\]
    for every $1\leq \ell \leq k$, but
    \[\mathfrak{M},w_{\sigma(0)} \Vdash \neg \langle \sigma[\rho] \rangle (\psi_1,\dots,\psi_k).\]
    We have now two cases based on whether or not $\sigma(0) = 0$. First, if $\sigma(0) = 0$, then since $\mathfrak{M},w_0 \Vdash \xi_3$, we have that
    \[\mathfrak{M},w_0 \Vdash \langle \rho \rangle (\psi_{\sigma^{-1}(1)}, \dots, \psi_{\sigma^{-1}(k)}) \to \langle \sigma [\rho] \rangle (\psi_1,\dots,\psi_k).\]
    By assumption, $(w_0,w_1,\dots,w_k) \in \rho^\mathfrak{M}$ and $\mathfrak{M},w_{\sigma(\sigma^{-1}(\ell))} \Vdash \psi_{\sigma^{-1}(\ell)}$, for every $1\leq \ell \leq k$, and hence
    \[\mathfrak{M},w_0 \Vdash \langle \rho \rangle (\psi_{\sigma^{-1}(1)},\dots,\psi_{\sigma^{-1}(k)}),\]
    which implies that
    \[\mathfrak{M},w_0 \Vdash \langle \sigma[\rho] \rangle (\psi_1,\dots,\psi_k),\]
    a contradiction.
    
    Consider then the case $\sigma(0) \neq 0$. Since $\mathfrak{M},w_0 \Vdash \xi_2$, we have that
    \[\mathfrak{M},w_0 \Vdash \bigg(\neg \psi_{\sigma^{-1}(0)} \lor \neg \langle \rho \rangle (\neg \psi_{\sigma^{-1}(1)},\dots,\underbrace{\langle \sigma[\rho] \rangle(\psi_1,\dots,\psi_k)}_{\sigma(0)\text{:th formula}},\dots,\neg \psi_{\sigma^{-1}(k)}) \bigg)\]
    By assumption $\mathfrak{M},w_{\sigma(\sigma^{-1}(0))} \Vdash \psi_{\sigma^{-1}(0)}$. Furthermore, since $(w_0,w_1,\dots,w_k) \in \rho^\mathfrak{M}$ and $\mathfrak{M},w_{\sigma(\sigma^{-1}(\ell))} \Vdash \psi_{\sigma^{-1}(\ell)}$, for every $\ell \neq \sigma(0)$, we must have that
    \[\mathfrak{M},w_{\sigma(0)} \Vdash \langle \sigma[\rho] \rangle (\psi_1,\dots,\psi_k),\]
    which is a clear contradiction.
\end{proof}

Now we will extend the model $\mathfrak{M}$ as follows.
\begin{enumerate}
    \item For every $2 \leq k \leq m, \rho \in \mathfrak{T}_k, \sigma \in S_k$ and $(w_1,\dots,w_k) \in \rho^\mathfrak{M}$, we will add $(w_{\sigma(1)},\dots,w_{\sigma(k)})$ to $(\sigma[\rho])^\mathfrak{M}$.
    \item For every $2 \leq k \leq m$ and $(w_1,\dots,w_k) \in W^k$, for which there does not exists $\rho \in \mathfrak{T}_k$ so that $(w_1,\dots,w_k) \in \rho^\mathfrak{M}$, we let $\rho$ denote the relation promised by Lemma \ref{lemma:pmlm_main_lemma} and we will add the tuple $(w_{\sigma(1)},\dots,w_{\sigma(k)})$ to $(\sigma[\rho])^\mathfrak{M}$, for every $\sigma \in S_k$.
\end{enumerate}
We will still use $\mathfrak{M}$ to denote the resulting model.

\begin{lemma}
    $\mathfrak{M},w \Vdash t(\varphi)$.
\end{lemma}
\begin{proof}
    A routine induction.
\end{proof}

Now we are ready to use $\mathfrak{M}$ to construct a model for $\varphi$. We define a Kripke model $\mathfrak{N} = (W^*,(R^\mathfrak{N})_{R\in\tau},V^*)$ over $\tau$ as follows. First, we define that 
\[W^* := W \times \{2,\dots,m\} \times \mathbb{N},\]
where $\mathbb{N}$ is the set of natural numbers. In what follows we will adopt the convention that we will associate to every $k$-table $\rho$ an unique index from the set $\mathbb{N}$, which we simply denote by $\rho$. We start our model construction by specifying that for every $(w,\ell,r) \in W^*$ we have that $(w,\ell,r) \in V^*(p)$ iff $w\in V(p)$. Next we will assign tables to tuples. We first define that for every $(w_0,\ell,r) \in W^*$ and $(w_0,w_1,\dots,w_k) \in \rho^\mathfrak{M}$ the tuple
\[((w_0,\ell,r),(w_1,2,r + \rho), \dots, (w_k,k + 1,r + \rho))\]
realizes the type $\rho$. Observe that each such tuple consists of $k+1$ distinct elements. Indeed, the first element is distinct from the remaining elements because $r \neq r + \rho$, while the remaining elements in the tuple are pairwise distinct because they differ with respect to their second coordinate.

Notice that if we force a tuple to realize a table $\rho$, then for every $\sigma \in S_{k+1}$ the permutation of this tuple under $\sigma$ realizes the type $\sigma[\rho]$. Hence, it is not obvious that the above procedure does not assign different tables to some tuples.

\begin{claim}
    In the above procedure, no tuple is assigned a table more than once.
\end{claim}
\begin{proof}
    Suppose that we have assigned tables $\rho_1$ and $\rho_2$ to a tuple $(w_0,w_1,\dots,w_k)$. We want to show that $\rho_1 = \rho_2$. By construction, our assumption implies that there are tuples
    \[((w_0',\ell',r'),(w_1',2,r' + \rho'),\dots,(w_k',k+1,r' + \rho'))\]
    and
    \[((w_0'',\ell'',r''),(w_1'',2,r'' + \rho''),\dots,(w_k'',k+1,r'' + \rho''))\]
    and permutations $\sigma_1$ and $\sigma_2$ so that $\sigma_1$ (respectively $\sigma_2$) applied to the first (respectively the second) tuple gives $(w_0,w_1,\dots,w_k)$, and furthermore that $\sigma_1[\rho'] = \rho_1$ and $\sigma_2[\rho''] = \rho_2$. Since for every $2 \leq \ell \leq k + 1$ there exists an unique element in the tuple $(w_0,w_1,\dots,w_k)$ which has $\ell$ as its second coordinate, we must have that the two above tuples are in fact the same tuples, since they are both permutations of $(w_0,w_1,\dots,w_k)$. In particular, $\rho' = \rho''$, since $r' = r''$. Finally, because this single tuple consists of $k+1$ distinct elements and permutating it with either $\sigma_1$ or $\sigma_2$ gives the same result -- namely $(w_0,w_1,\dots,w_k)$ -- we must have that $\sigma_1 = \sigma_2$, and hence $\rho_1 = \rho_2$.
\end{proof}

Finally, for every tuple $((w_0,\ell_0,r_0),\dots,(w_k,\ell_k,r_k)$ for which we have not yet assigned a table, we will pick some $\rho \in \mathfrak{T}_{k+1}$ for which $(w_0,\dots,w_k) \in \rho^\mathfrak{M}$ and assign the corresponding table to our tuple. This completes the definition of $\mathfrak{N}$.

\begin{lemma}
    For every $\psi \in \subf(\varphi)$ and $w_0\in W$ we have that
    \[\mathfrak{N},(w_0,\ell,r) \Vdash \psi \Leftrightarrow \mathfrak{M},w_0 \Vdash t(\psi).\]
\end{lemma}
\begin{proof}
    A routine induction.
\end{proof}

In particular $\mathfrak{N}$ is a model of $\varphi$ and hence $\varphi$ is satisfiable. Thus we can conclude that $\varphi$ is satisfiable iff $\Theta$ is.

\section{Model checking problem of $\PML(p,s,\neg,\cap)$}\label{sec:model_checking}

In this section we prove that the combined complexity of $\PML(p,s,\neg,\cap)$ is \textsc{PTime}-complete. Note that the corresponding lower bound follows already from the fact that the combined complexity of standard modal logic is \textsc{PTime}-complete \cite{GradelSurvey}.

We start by defining precisely how we will encode Kripke models. (We will assume that their domains are equipped with some arbitrary linear order). In fact, we will describe how the encode arbitrary relational models, since it avoids some notational clutter. Given two strings $x$ and $y$, we will use $x \# y$ to denote their concatenation. The \emph{list encoding} of $\mathfrak{A}$ is the sequence
\[1^{|A|} \rhd \lenc(R_1) \rhd \dots \rhd \lenc(R_m)\]
where $\rhd$ is a separator character (the use of which could be of course avoided) and each $\lenc(R_i)$ is a sequence
\[r_1\# r_2\# \dots \# r_{|R_i|},\]
where each $r_j$ is a sequence consisting of $ar(R_i)$-many binary strings of length $\log_2(|A|)$. We note that the length of the list encoding of a model $\mathfrak{A}$, denoted by $||\mathfrak{A}||$, is
\[O\bigg(|A|+\sum_{1\leq i\leq m} |R_i|ar(R_i)\log_2(|A|)\bigg).\]

\begin{remark}
    The encoding of models that we have presented here is not the only encoding of relational structures one encounters in the literature. Another standard choice of encoding can be found in \cite{Libkin04}, where the encoding essentially describes the ``adjacency" matrix of each relation, i.e., for every relation $R$ and every tuple of length $ar(R)$ there is a single bit which indicates whether that tuple belongs to the interpretation of $R$. If this encoding of models is used, then the \textsc{PTime} upper bound on the model checking problem of $\PML(p,s,\neg,\cap)$ becomes somewhat trivial, since one can compute complements of relations in linear time.
\end{remark}

Next we will present our model checking algorithm. As an important preliminary step, the following lemma will allow us to restrict our attention to formulas that contain only terms which use negation in a very restricted way.

\begin{lemma}\label{lemma:eliminate_negation}
    Suppose that $\mathcal{R} \in \gra(p,s,\neg,\cap)[\tau]$. We can compute in polynomial time a term $\mathcal{R}' \in \gra(p,s,\neg,\backslash,\cap,\cup)$ so that $\mathcal{R}$ is equivalent with $\mathcal{R}'$ and $\mathcal{R}'$ is either of the form $\mathcal{R}''$ or of the form $\neg \mathcal{R}''$, for some $\mathcal{R}'' \in \gra(p,s,\backslash,\cap,\cup)[\tau]$.
\end{lemma}
\begin{proof}
    Using Lemma \ref{lemma:simple_equations}, we can bring all the negations occurring in the input term $\mathcal{R}$ to the start of the term, which -- after eliminating consecutive negations -- results in a term which is either of the form $\mathcal{R}'$ or $\neg \mathcal{R}'$, where $\mathcal{R}' \in \mathrm{GRA}(p,s,\backslash,\cap,\cup)[\tau]$.
\end{proof}

Suppose now that $\varphi \in \PML(p,s,\neg,\cap)[\tau]$ and $\mathfrak{M} = (W,(R^\mathfrak{M})_{R\in \tau},V)$. Our goal is to compute the set of worlds in $\mathfrak{M}$ where $\varphi$ is true. By applying Lemma \ref{lemma:eliminate_negation}, we can assume that in each subformula $\langle \mathcal{R} \rangle (\psi_1,\dots,\psi_k)$ the term $\mathcal{R}$ is either of the form $\mathcal{R}'$ or $\neg \mathcal{R}'$, where $\mathcal{R}' \in \gra(p,s,\backslash,\cap,\cup)[\tau]$. Using induction, one can show that the size of $\llbracket \mathcal{R}' \rrbracket_{\mathfrak{M}}$ is only polynomial with respect to $||\mathfrak{M}||$.

Now we will describe the model checking algorithm, which extends the standard labeling algorithm that is often used in the context of modal logics \cite{Vardi1996WhyIM}. The algorithm will use some enumeration of the subformulas $\varphi_1,\dots,\varphi_n$ of $\varphi$ which satisfies the requirement that if $\varphi_j$ is a proper subformula of $\varphi_i$, then $j < i$.

\begin{enumerate}
    \item [] $\nu = \varnothing$
    \item [] \textbf{for} $i=1$ \textbf{through} $n$ \textbf{do}:
    \begin{itemize}
        \item [] \textbf{if} $\varphi_i = p$ \textbf{then} $\nu := \nu \cup \{(p,V(p))\}$
        \item [] \textbf{if} $\varphi_i = \neg \varphi_j$ \textbf{then} $\nu := \nu \cup \{(\varphi_i,W\backslash \nu(\varphi_j)\}$
        \item [] \textbf{if} $\varphi_i = \varphi_j \land \varphi_k$ \textbf{then} $\nu := \nu \cup \{(\varphi_i,\nu(\varphi_j) \cap \nu(\varphi_k)\}$
        \item [] \textbf{if} $\varphi_i = \langle \mathcal{R} \rangle (\varphi_{i_1},\dots,\varphi_{i_k})$ \textbf{then}
        \[\nu := \nu \cup \bigg(\varphi_i, \bigg\{w\in W \bigg\vert \bigg(w \times \nu(\varphi_{i_1}) \times \dots \times \nu(\varphi_{i_k}) \bigg) \cap \llbracket \mathcal{R} \rrbracket_\mathfrak{M} \neq \varnothing\bigg\}\bigg)\]
        \item [] \textbf{if} $\varphi_i = \langle \neg\mathcal{R} \rangle (\varphi_{i_1},\dots,\varphi_{i_k})$ \textbf{then}
        \begin{enumerate}
            \item [] $U := \varnothing$
            \item [] \textbf{for} $w\in W$ \textbf{do}:
            \begin{itemize}
                \item [] \textbf{for} $(w_1,\dots,w_k) \in \nu(\varphi_{i_1}) \times \dots \times \nu(\varphi_{i_k})$ \textbf{do}:
                \begin{enumerate}
                    \item [] \textbf{if} $(w,w_1,\dots,w_k) \not\in \llbracket \mathcal{R} \rrbracket_\mathfrak{M}$ \textbf{then} $U := U \cup \{w\}$ and \textbf{break}
                \end{enumerate}
            \end{itemize}
            \item [] $\nu := \nu \cup \{(\varphi_i,U)\}$
        \end{enumerate}
    \end{itemize}
    \item [] \textbf{return} $\nu(\varphi)$
\end{enumerate}

\noindent It is straightforward to check that the above algorithm is correct. We note that in certain places the description of the algorithm is, for ease of exposition, somewhat informal. In particular, we have not specified how in the case of $\langle \neg \mathcal{R} \rangle (\varphi_{i_1},\dots,\varphi_{i_k})$ the for-loop going through the set $\nu(\varphi_{i_1}) \times \dots \nu(\varphi_{i_k})$ should be implemented; this is in fact a somewhat important detail, since we can not always construct this set explicitly, since in the worst case its size is proportional to $|W|^{|\varphi|}$. However, it is clear that this explicit construction can be avoided, because we can alternatively maintain $k \leq |\varphi|$ pointers to elements of the sets $\nu(\varphi_{i_\ell})$.

We are now left with the easy task of proving that our algorithm runs in polynomial time.

\begin{lemma}\label{lemma:algorithm_runs_in_polynomial_time}
    The above algorithm runs in time polynomial with respect to 
    $$|\varphi|\times || \mathfrak{M} ||.$$
\end{lemma}
\begin{proof}
    The outer most for-loop is executed $|\subf(\varphi)|\leq |\varphi|$ times, so it suffices to argue that each case within the loop can be done in time polynomial with respect to $|\varphi| \times ||\mathfrak{M}||$. The most non-trivial case is the case of $\langle \neg \mathcal{R} \rangle (\varphi_{i_1},\dots,\varphi_{i_k})$, where we can make the simple observation that the running time of the for-loop going through $\nu(\varphi_{i_1}) \times \dots \nu(\varphi_{i_k})$ is bounded above by the size of $\llbracket \mathcal{R} \rrbracket_{\mathfrak{M}} \leq || \mathfrak{M} ||$, since it stops after we have encountered a tuple which does not belong to $\llbracket \mathcal{R} \rrbracket_\mathfrak{M}$ (or after we have went through all the relevant $k$-tuples).
\end{proof}

Since the model checking problem of standard modal logic is \textsc{PTime}-complete, we have the desired result.

\begin{theorem}\label{thm:mc_for_pbml}
    The model checking problem of $\PML(p,s,\neg,\cap)$ is \textsc{PTime}-complete.
\end{theorem}

\section{Conclusions}\label{sec:conclusions}

We have studied the computational complexity of model checking and satisfiability problems of polyadic modal logics extended with permutations and Boolean operators on accessibility relations. Concerning satisfiability problems, we have proved that the satisfiability problems of both polyadic modal logic extended with negations of accessibility relations $\PML(\neg)$ and full polyadic Boolean modal logic extended with permutations over $\PML(p,s,\neg,\cap)$ are \textsc{ExpTime}-complete, the latter under the assumption that the underlying set of accessibility relations is finite, which is necessary if \textsc{ExpTime} $\neq$ \textsc{NExpTime}. We have also established that the model checking problem for full polyadic Boolean modal logic extended with permutations $\PML(p,s,\neg,\cap)$ is \textsc{PTime}-complete. Our results contribute to the research program outlined in \cite{IsoTuisku2021DescriptionLA} and extend the results of \cite{ComplexityReasoningBooleanModalLogics, Lutz2001ModalLA} to polyadic context.

Concerning future research directions, the reductions that we used in establishing complexity bounds on satisfiability problems seem to be quite robust, and hence we expect that in the future they can be used to extend the results presented here. For instance, one can most likely show that $\PML(p,s,\neg)$ has an \textsc{ExpTime}-complete satisfiability problem, which we have not yet been able to do. In this direction a natural intermediate problem would be to establish that $\PML(p,s) + \E$ has an \textsc{ExpTime}-complete satisfiability problem, since then one might be able to adapt the techniques used in this paper to reduce the satisfiability problem of $\PML(p,s,\neg)$ to that of $\PML(p,s) + \E$.

Another obvious direction would be to show that if the underlying set of accessibility relations is finite, then the satisfiability problem of $\PML(p,s,\neg,\cap)$ extended with equality $\PML(I,p,s,\neg,\cap)$ is \textsc{ExpTime}-complete. Indeed, such a result would fully generalize the main result of \cite{Lutz2001ModalLA} to polyadic context, which states the satisfiability problem of $\mathrm{ML}(I,s,\neg,\cap)$ is \textsc{ExpTime}-complete, when the underlying set of accessibility relations is finite.

\bibliography{csl}
\end{document}